\documentclass[12pt, letterpaper]{article} 

\usepackage{amsfonts}
\usepackage{graphicx}

\usepackage{amsmath}
\usepackage{amssymb}
\usepackage{hyperref}
\usepackage{amsthm}
\usepackage[linesnumbered,lined,algo2e]{algorithm2e}
\usepackage{booktabs}
\usepackage{multirow}

\newtheorem{definition}{Definition}
\newtheorem{remark}{Remark}
\newtheorem{problem}{Problem}
\newtheorem{lemma}{Lemma}
\newtheorem{theorem}{Theorem}
\newtheorem{proposition}{Proposition}
\newtheorem{corollary}{Corollary}
\newtheorem{assumption}{Assumption}

\newcommand{\bartt}[1]{#1'}

\title{An indirect computational procedure for receding horizon hybrid optimal control}

\author{
Babak~Tavassoli
\\
{\small K. N. Toosi University of Technology}
\\
{\small E-mail: tavassoli@kntu.ac.ir}
}

\hypersetup{
  pdftitle={},
  pdfauthor={}
}
\date{}

\begin{document}

\maketitle

\begin{abstract}
In this work, solution of the finite horizon hybrid optimal control problem as the central element of the receding horizon optimal control (model predictive control) is investigated based on the indirect approach. 
The response of a hybrid system within the prediction horizon is composed of both discrete-valued sequences and continuous-valued time-trajectories. Given a cost functional, the optimal continuous trajectories can be calculated given the discrete sequences by the means of the recent results on the hybrid maximum principle. It is shown that these calculations reduce to solving a system of algebraic equations in the case of affine hybrid systems. Then, a branch and bound algorithm is proposed which determines both the discrete and continuous control inputs by iterating on the discrete sequences. It is shown that the algorithm finds the correct solution in a finite number of steps if the selected cost functional satisfies certain conditions. Efficiency of the proposed method is demonstrated during a case study through comparisons with the main existing method.
\end{abstract}


\section{Introduction}
Hybrid systems include both discrete-valued and continuous-valued state variables that interact with each other \cite{R:HOP1, R:HHB, R:HS2, R:HG1}. A system with only discrete states can be described by an automaton, while a system with only continuous states can be described using differential equations. However, the interaction between the two types of state results in serious complexities in the case of hybrid systems. Considerable research has been devoted to cope with these complexities due to the increase of applications with hybrid dynamical nature in industries, energy systems, biological systems, and more generally, in the cyber-physical systems \cite{R:CPS1, R:CPS2, R:Moness2020, R:CPS3, R:HHB}.

There are different approaches to hybrid system problems. They range from the extensions of the methods for discrete systems in the computer science to the methods that extend the ideas in control theory for continuous systems. Some examples are the applications of the Lyapunov or small gain theorems \cite{R:HG1, R:HG3, R:CPS3}, verification and control of hybrid systems based on model checking or symbolic modeling \cite{R:HS2, R:Reissig2019, R:Goubault2019}, optimal control of hybrid systems \cite{R:HO1, R:HO2, R:HOP1}, and model predictive control (MPC) \cite{R:HMPC1, R:HMPC4, R:HMPC2, R:Moness2020}.
The MPC or more precisely the receding horizon optimal control method, solves a finite horizon optimal control problem at every time step in order to compute the control signals. 
Formal extension of the MPC method to hybrid systems was started in \cite{R:MLD} which applies a direct approach to solution of the optimal control problem. The direct approach is based on approximating the system response by functions of time with a finite number of parameters which reduces the optimal control problem to an optimization with a finite number of decision variables. The larger is the number of parameters, the more accurate is the approximation, but also the heavier is the load of computations. The hybrid MPC approach of \cite{R:MLD} and its following works like \cite{R:HMPC1, R:HMPC4, R:Frick2019, R:Moness2020} uses time-discretization to reduce the MPC problem to a mixed integer program (MIP) over the system variables during the prediction horizon. 

In this work, an MPC algorithm is developed based on the indirect approach to solution of the finite horizon hybrid optimal control based on the hybrid maximum principle (HMP) \cite{R:NM1, R:NM2, R:HO1}. The two main difficulties in this regard are: 1- solving the differential-algebraic system of equations given by the HMP, and 2- finding the optimal sequences of discrete state and discrete input that are assumed to be known in the HMP.
The first difficulty is tackled by reducing the HMP equations to an algebraic system of equations in terms of only the jump times within the prediction horizon for the special case of affine hybrid systems with quadratic cost functionals. 
Also, the second difficulty is addressed by proposing a branch-and-bound algorithm to compute the optimal discrete sequences iteratively. It is proved that the algorithm finds the correct solution in a finite number of steps either if the cost functional assigns cost to the jumps or the number of jumps within the prediction horizon is restricted. 
The proposed indirect computational procedure does not apply approximations which is an advantage. More importantly, the computational complexity of the indirect approach is less than the direct approach in general. 
The reason is that the proposed indirect approach is based on solving equations with a few unknown variables that are primarily the jump time instants. But, in the indirect approach a much larger set of decision variables must be defined for each sampling instant (the sampling period is typically an order of magnitude smaller than the average time interval between jumps).
To demonstrate these advantages, a case study is provided in which the proposed method is applied to the hybrid system benckmark in \cite{R:SU2} and comparisons are made with the hybrid MPC approach of \cite{R:MLD,R:HMPC4,R:HMPC1}. 

The organization of the paper is as follows. The required definitions followed by the hybrid MPC problem statement are provided in Section~\ref{S:HS}. The proposed indirect MPC approach including the utilization of the HMP and the algorithm for calculation of the discrete sequences is presented in Section~\ref{S:MPC}. Correctness and finiteness of the proposed algorithm are studied in Section~\ref{S:ALP}. Some issues that cannot be deeply investigated in this work are pointed out in Section~\ref{S:RX}. The case study is provided in Section~\ref{S:EXM} and conclusions are made at the end.

\textit{Notation}: The sets of real numbers and integers are denoted by $\mathbb{R}$ and $\mathbb{Z}$ respectively (the subsets of non-negative or positive numbers are denoted by adding $\geq 0$ or $>0$ in the subscript). Floor of $a\in\mathbb{R}$ is denoted as $\lfloor a\rfloor$.
For a function $F:A\to B$, the restriction of $F$ to $A'\subset A$ is denoted as $F|_{A'}:A'\to B$. 
The set of integers that are not less than $m\in\mathbb{Z}$ and not greater than $n\in\mathbb{Z}$ is denoted by $[m..n]$. The logical conjunction and disjunction operators are denoted by $\wedge$ and $\vee$ respectively.
For a real matrix $M\in\mathbb{R}^{m\times n}$, the element at row $i\in[1..m]$ and column $j\in[1..n]$ is denoted by $[M]_{i,j}$.
For an arbitrary set $B$, the set of all sequences of elements in $B$ indexed by $[m..n]\subset\mathbb{Z}$ is denoted by $B^{[m..n]}$. For $b\in B^{[m..n]}$, the element which corresponds to $i\in[m..n]$ is denoted by $b_i$ and the number of elements of $b$ is denoted by $|b| = n-m+1$. If $|b|=1$, then we simply write $b$ instead of $b_m = b_n$. The empty sequence is denoted by $\{\}$. 
For $b\in B^{[m..n]}$, a subsequence of $b$ is a sequence denoted as $b_{p..r}\in B^{[p..r]}$ for some $[p..r]\subset[m..n]$, such that $(b_{p..r})_i=b_i$ for all $i\in[p..r]$. It is said that $b\in B^{[m..n]}$ is a prefix of $b'\in B^{[m..r]}$ denoted as $b\prec b'$, if $r\geq n$ and $b'_{m..n}=b$. The concatenation of a sequence $b\in B^{[m..n]}$ and an element $b'\in B$ is a sequence in $B^{[m..n+1]}$ denoted as $b \, b'$ such that $(b \, b')_{m..n}=b$ and $(b \, b')_{n+1}=b'$.

\section{Hybrid MPC Problem}
\label{S:HS}
Before stating the hybrid MPC problem, hybrid systems and their time responses need to be defined in this section. 

\subsection{Definition of hybrid system}
There are various formal definitions of hybrid systems. The following definition is based on the notions of hybrid automaton in \cite{R:HHB}.

\begin{definition} 
\label{D:HS}
A hybrid system $\mathcal{H}$ is a tuple $(Q$, $\Sigma$, $\Theta$, $f$, $D$, $G$, $R)$ where
\begin{itemize}
\item $Q$ is a finite set of discrete state values,
\item $\Sigma$ is a finite set of discrete input values,
\item $\Theta\subseteq Q\times\Sigma\times Q$ is a set of transitions (jumps),
\item $D_q\in\mathbb{R}^{n_x}$ for every $q\in Q$ are domains,
\item $f_q:D_q\times\mathbb{R}^{n_u}\to \mathbb{R}^{n_x}$ for every $q\in Q$ are vector fields, 
\item $G_{q,\sigma,q'}\in\mathbb{R}^{n_x}$ for every $(q,\sigma,q')\in\Theta$ are guard sets,
\item $R_{q,\sigma,q'}:\mathbb{R}^{n_x}\to\mathbb{R}^{n_x}$ for every $(q,\sigma,q')\in\Theta$ are reset maps.
\end{itemize}
\end{definition}

Several of the problems in regard with hybrid systems are studied under the following assumption.

\begin{assumption}
\label{A:1}
Considering a hybrid system $\mathcal{H} = (Q$, $\Sigma$, $\Theta$, $f$, $D$, $G$, $R)$ according to the Definition~\ref{D:HS}, it is assumed that the functions $f$ and $R$ are differentiable. Also, for every $(q,\sigma,q')\in\Theta$ there exist a differentiable function $g_{q,\sigma,q'}:\mathbb{R}^{n_x}\to\mathbb{R}$ such that  
\begin{align}
G_{q,\sigma,q'} = \{x\in\mathbb{R}^{n_x} : g_{q,\sigma,q'}(x)\le 0\}. \label{E:PMP.10} 
\end{align}
\end{assumption}

In this paper we are particularly interested in the class of hybrid systems that can be described as an affine hybrid system in which the vector fields, reset maps, and the functions $g_{q,\sigma,q'}$ in (\ref{E:PMP.10}) take the affine forms
\begin{subequations}
\label{E:AHS}
\begin{align}
f_q(x,u)&=A_q x+B^u_q u+B^c_q, \label{E:AHS.1}\\
g_{q,\sigma,q'}(x)&=M^x_{q,\sigma,q'}x+M^c_{q,\sigma,q'}, \label{E:AHS.2}\\
R_{q,\sigma,q'}(x)&=L^x_{q,\sigma,q'}x+L^c_{q,\sigma,q'} \label{E:AHS.3}
\end{align}
\end{subequations}
for every $(q,\sigma,q')\in\Theta$ in which the matrix and vector coefficients have the appropriate dimensions. More general cases will be discussed in Subsection~\ref{SS:GPW}.

\subsection{Time response of hybrid system}
\label{SS:EX}
A change of the discrete state is referred to as a jump. An increasing sequence of time instants $t^s\in\mathbb{R}^{[0..n]}$ with $n\in\mathbb{Z}_{>0}$ can be defined such that $t^s_0$ is the initial time, $t^s_n$ is the final time, and $t^s_i$ for $i\in[1..n-1]$ are the jump instants. Both $n$ and $t^s_n$ can tend to infinity. For an arbitrary time dependent variable $y:[t^s_0,t^s_n]\to\mathbb{R}^{n_y}$ with $n_y\in\mathbb{Z}_{>0}$ the following notations are used.
\begin{subequations}
\begin{align}
y_i^- &= \ \underset{t \uparrow t^s_i}{\lim} \ y(t) && i\in[1..n] \label{E:1.1}\\
y_i^+ &= \ y(t^s_i) && i\in[0..n-1] \label{E:1.2}
\end{align}
\end{subequations}

The time response of the hybrid system which is denoted as an execution can be defined as in the following.

\begin{definition}
\label{D:EX}
An execution of a hybrid system $\mathcal{H} = (Q$, $\Sigma$, $\Theta$, $f$, $D$, $G$, $R)$ is a tuple $E = (t^s$, $q$, $\sigma$, $x$, $u)$ where 
\begin{itemize}
\item $t^s\in\mathbb{R}^{[0..n]}$ is the time sequence, 
\item $q\in Q^{[1..n]}$ is the discrete state sequence, 
\item $\sigma\in\Sigma^{[1..n-1]}$ is the discrete input sequence,
\item $x:[t^s_0,t^s_n]\to\mathbb{R}^{n_x}$ is the continuous state trajectory, 
\item $u:[t^s_0,t^s_n]\to\mathbb{R}^{n_u}$ is the continuous input trajectory,
\end{itemize}
\noindent
for some $n\in\mathbb{Z}_{>0}$, such that $t^s_i>t^s_{i-1}$ for $i\in[1..n-1]$, $t^s_n\geq t^s_{n-1}$, relations (\ref{E:2.1}) and (\ref{E:2.5}) hold for $t\in[t^s_{i-1},t^s_i)$ with $i\in[1..n]$, and relations (\ref{E:2.4}) through (\ref{E:2.2}) hold for $i\in[1..n-1]$. 
\begin{subequations}
\label{E:2}
\begin{align}
&\dot{x}(t)=f_{q_i}(x(t),u(t)) \label{E:2.1}\\
&x(t)\in D_{q_i} \label{E:2.5}\\
&(q_i,\sigma_i,q_{i+1})\in\Theta \label{E:2.4} \\
&x_i^-\in G_{q_i,\sigma_i,q_{i+1}} \label{E:2.3}\\
&x_i^+=R_{q_i,\sigma_i,q_{i+1}}(x_i^-) \label{E:2.2} 
\end{align}
\end{subequations}
\end{definition}

The set of all executions of a hybrid system $\mathcal{H}$ is denoted by $\mathcal{E(H)}$. The set of executions that satisfy $t^s_0=0$, and $x(0) = x_{ic}$ is denoted by $\mathcal{E(H}$, $x_{ic})$. Also, $\mathcal{E(H}$, $x_{ic}$, $T)$ denotes the set of executions in $\mathcal{E(H}$, $x_{ic})$ that satisfy $t^s_n=T$.

During the time interval $[t^s_{i-1},t^s_i)$ with $i\in[1..n]$, the discrete state is $q_i$, and the continuous state evolves according to (\ref{E:2.1}). This type of evolution of the state is denoted as a flow. The time instant for the $i$th jump $t^s_i$ for $i\in[1..n-1]$ is determined as the time at which $x(t)$ reaches the boundary of $G_{q_i,\sigma_i,q_{i+1}}$ according to (\ref{E:2.3}). In this way, we avoid a kind of uncertainty when both flow and jump are possible at the same time by giving priority to jumps. At the time instant of jump, the continuous state is reset according to (\ref{E:2.2}). 

It is said that an execution $E = (t^s$, $q$, $\sigma$, $x$, $u)$ with $|q| = n$ is a prefix of another execution $E' = ({t^s}'$, $q'$, $\sigma'$, $x'$, $u')$, if we have $q\prec q'$, $\sigma\prec\sigma'$, $t^s_{0..n-1} = t^{s\prime}_{0..n-1}$, $t^s_n\leq {t^s_n}'$, $x = x'|_{[t^s_0,t^s_n]}$, and $u = u'|_{[t^s_0,t^s_n]}$.

To avoid confusion, it is mentioned that the notion of execution of a hybrid system is apart from the concept of {\it controller execution} which will be used to denote a run of the control algorithm.

\subsection{The hybrid MPC problem statement}
\label{SS:PB}

The MPC algorithm solves an optimal control problem at every time step over a finite horizon which starts from the current time $t$ and ends at $t+T_h$ in future. Then, the part of the calculated input which corresponds to the current time is applied to the plant and the rest of the calculated values are neglected. This procedure is repeated with a controller execution period of $T_c$ in order to achieve a desirable control performance. The interval $[t,t+T_h]$ (or sometimes its length $T_h$) is referred to as the prediction horizon.

Considering an execution $E = (t^s$, $q$, $\sigma$, $x$, $u)$, the cost functional $J$ to be minimized is defined as in the following. 
\begin{subequations}
\label{E:JC}
\begin{align}
J(E) =& J_m(E) + h^e_{q_n}(x^-_n) \label{E:JC.2} \\[6pt]
J_m(E) =& \sum_{i=1}^n \int_{t^s_{i-1}}^{t^s_i} l_{q_i}(x(t),u(t))dt \ + \sum_{i=1}^{n-1} h_{i,q_i,\sigma_i,q_{i+1}}(x^-_i) \label{E:JC.1} 
\end{align}
\end{subequations}

In the above definition, $l_q:\mathbb{R}^{n_x}\times\mathbb{R}^{n_u}\to\mathbb{R}_{\geq 0}$, $h_{i,q,\sigma,q'}:\mathbb{R}^{n_x}\to\mathbb{R}_{\geq 0}$, and $h^e_q:\mathbb{R}^{n_x}\to\mathbb{R}_{\geq 0}$ are differentiable functions for every $i\in\mathbb{Z}_{>0}$, $(q,\sigma,q')\in\Theta$. 

In this work, we are particularly interested in a cost functional with quadratic elements as below (the matrix coefficients have the appropriate dimensions).
\begin{subequations}
\label{E:QJ}
\begin{align}
l_q(x,u)&={\scriptstyle\frac{1}{2}}[(x-\bar{x}_q)^T W^x_q (x-\bar{x}_q) + (u-\bar{u}_q)^T W^u_q (u-\bar{u}_q) + W^c_q] \label{E:QJ.1}\\
h_{i,q,\sigma,q'}(x)&=W^j_{i,q,\sigma,q'}+(x-\bar{x}_q)^T W^{jx}_{i,q,\sigma,q'} (x-\bar{x}_q) \label{E:QJ.2}\\
h^e_q(x)&={\scriptstyle\frac{1}{2}}(x-\bar{x}_q)^T W^e_q (x-\bar{x}_q) \label{E:QJ.3}
\end{align}
\end{subequations}

The optimal control problem that should be solved by the MPC algorithm at the time $t$ is as the Problem~\ref{PB:1} in the following. 

\begin{problem}
\label{PB:1}
Given a hybrid system $\mathcal{H} = (Q$, $\Sigma$, $\Theta$, $f$, $D$, $G$, $R)$ according to the Definition~\ref{D:HS}, a cost functional $J$ as in (\ref{E:JC}), prediction horizon $T_h\in\mathbb{R}_{>0}$, initial states $q_{ic}\in Q$ and $x_{ic}\in\mathbb{R}^{n_x}$, find an execution $E = (t^s$, $q$, $\sigma$, $x$, $u)\in$ $\mathcal{E(H}$, $x_{ic}$, $T_h)$ with $q_1=q_{ic}$ which minimizes $J$. 
\end{problem}

Due to the time invariance of the hybrid dynamics in (\ref{E:2}), the current time is shifted to the origin for simplicity such that $E\in\mathcal{E(H}$, $x_{ic}$, $T_h)$. The values of $x_{ic}$ and $q_{ic}$ must be respectively set to the values of continuous and discrete states at the current time $t$. After solving the problem and obtaining $E$, the continuous input $u(0)=u_0^+$ and the discrete input $\sigma_1$ must be applied to the hybrid system $\mathcal{H}$ as the plant. 
At an instant between two runs of the MPC algorithm at $t$ and $t+T_c$ denoted by $t'\in[t,t+T_c]$, one can alternatively apply $u(t'-t)$ and the calculated discrete input at $t'-t$ instead of $u(0)$ and $\sigma_1$ in order to improve accuracy.

\section{Hybrid MPC Algorithm}
\label{S:MPC}
This section, aims to develop the indirect hybrid MPC algorithm that should be run at each time step. The HMP is used in order to solve the underlying optimal control problem. It will be assumed that the feedback from the state variables is available.

\subsection{Calculating continuous trajectories given discrete sequences}
By fixing the discrete sequences of the executions in Problem~\ref{PB:1}, the Problem~\ref{PB:2} in the following is obtained. 

\begin{problem}
\label{PB:2}
Given a hybrid system $\mathcal{H} = (Q$, $\Sigma$, $\Theta$, $f$, $D$, $G$, $R)$ according to the Definition~\ref{D:HS}, a cost functional $J$ as in (\ref{E:JC}), prediction horizon $T_h\in\mathbb{R}_{>0}$, initial continuous state $x_{ic}\in\mathbb{R}^{n_x}$, and sequences $q\in Q^{[1..n]}$, $\sigma\in\Sigma^{[1..n-1]}$ for some $n\in\mathbb{Z}_{>0}$, find an execution $E\in$ $\mathcal{E(H}$, $x_{ic}$, $T_h)$ with discrete state sequence $q$ and discrete input sequence $\sigma$ which minimizes $J$. 
\end{problem}

Since jump is assumed to have priority with respect to flow, a jump occurs if the continuous state reaches the boundary of the corresponding guard set. Hence, (\ref{E:2.3}) together with (\ref{E:PMP.10}) results in
\begin{equation}
\label{E:GB}
g_{q_i,\sigma_i,q_{i+1}}(x_i^-)=0 \qquad i\in[1..n-1]
\end{equation}

The HMP has been presented in various forms in the previous works (see for example \cite{R:HO3, R:OC1, R:HO2, R:HO1, R:HO4}). The one which is more useful in here is provided in \cite{R:HO2} that can be represented as below.

\begin{proposition}
\label{TH:1}
Given a hybrid system $\mathcal{H} = (Q$, $\Sigma$, $\Theta$, $f$, $D$, $G$, $R)$ which satisfies the Assumption~\ref{A:1}, if an execution of $\mathcal{H}$ denoted by $E = (t^s$, $q$, $\sigma$, $x$, $u)$ solves the Problem~\ref{PB:2} for given sequences $q$ and $\sigma$, then there exist $\alpha_i\in\mathbb{R}$ for $i\in[1..n-1]$ with $n=|q|$ and $\lambda:[t^s_0,t^s_n]\to\mathbb{R}^{n_x}$ denoted as the costate such that the set of equations (\ref{E:PMPc}) for $i\in[1..n]$, (\ref{E:PMPj}) for $i\in[1..n-1]$, and (\ref{E:PMPf.1}) are satisfied with the Hamiltonian function $H$ defined in (\ref{E:PMPh}).
\end{proposition}
\begin{subequations}
\label{E:PMPc}
\begin{align}
&\dot{\lambda}(t) = - D^T_x H_{q_i}(x,u,\lambda) \quad t\in[t^s_{i-1}, t^s_i)~, \label{E:PMPc.1}\\
&H_{q_i}(x,\lambda,u) \leq H_{q_i}(x,\lambda,w) \quad \forall w\in\mathbb{R}^{n_u}~, \label{E:PMPc.2}
\end{align}
\end{subequations}

\begin{subequations}
\label{E:PMPj}
\begin{align}
&\lambda_i^- = D^T_x R_{q_i}(\sigma_i,q_{i+1},x_i^-) \ \lambda_i^+ + 
D^T_x h_{i,q_i,\sigma_i,q_{i+1}}(x_i^-) + \alpha_i D^T_x g_{q_i,\sigma_i,q_{i+1}}(x_i^-)~, \label{E:PMPj.1} \\[6pt]
&H_{q_i}(x_i^-,\lambda_i^-,u_i^-) = H_{q_{i+1}}(x_i^+,\lambda_i^+,u_i^+)~, \label{E:PMPj.2} 
\end{align}
\end{subequations}
\begin{align}
&\lambda_{n}^- = D^T_x h_{f,q_n}(x_{n}^-)~, \label{E:PMPf.1} 
\end{align}
\begin{equation}
\label{E:PMPh}
H_q(x,u,\lambda) = l_q(x,u) + \lambda^T f_q(x,u) ~.
\end{equation}

In the above equations, $D^T_x$ denotes the transpose of the Jacobian matrix with respect to $x$ which becomes the gradient column vector for scalar-valued functions.

The conditions given in the Proposition~\ref{TH:1} together with (\ref{E:2.1}), (\ref{E:2.2}), and (\ref{E:GB}) constitute a differential-algebraic system of equations that can be solved for $x$, $u$, and $t^s_i$ for $i\in[1..n-1]$ (in order to solve the Problem~\ref{PB:2}). In general, the solution can be obtained by using the numerical methods for hybrid optimal control based on the HMP \cite{R:NM1, R:NM2, R:HO1}. However, the solution process becomes considerably easier for the class of affine hybrid systems as explained in the next part.

\begin{remark}
\label{RM:1}
A special type of jump which is sometimes referred to as controlled switching \cite{R:HO1, R:HO2}, is when $G_{q,\sigma,q'}=\mathbb{R}^{n_x}$ for some $(q,\sigma,q')\in\Theta$. In this case, the controller is free to make the jump $(q,\sigma,q')$ at every time instant in which the discrete state is $q$ and the discrete input is $\sigma$. For this purpose, $g_{q,\sigma,q'}$ can be defined to be zero at every point such that (\ref{E:GB}) is always satisfied. In the case of affine hybrid systems, the matrices $M^x_{q,\sigma,q'}$ and $M^c_{q,\sigma,q'}$ in (\ref{E:AHS.2}) are set to zero matrices.
\end{remark}

\subsection{The case of affine hybrid systems}
\label{SS:AHS}
The Problem~\ref{PB:2} can be solved much more efficiently in the case of affine hybrid systems with the cost functional (\ref{E:JC}) which has quadratic elements in the form of (\ref{E:QJ}).

Minimization of $H$ according to (\ref{E:PMPc.2}) gives 
\begin{equation}
u=\bar{u}_{q_i} - W^{u^{-1}}_{q_i} {B^u_{q_i}}^T\lambda \qquad \forall \ i\in[1..n] \label{E:US}
\end{equation}

By replacing $u_i$ from the above equation in (\ref{E:2.1}) and (\ref{E:PMPc.1}) for the affine case in (\ref{E:AHS}) and (\ref{E:QJ}), we have two coupled differential equations that can be written as the following for $i\in[1..n]$.
\begin{subequations}
\label{E:LPr}
\begin{align}
&\frac{d}{dt}\begin{bmatrix} x \\ \lambda\\ 1 \end{bmatrix} = A^e_{q_i} \begin{bmatrix} x \\ \lambda\\ 1 \end{bmatrix}, \quad t\in[t^s_{i-1}, t^s_i) \label{E:LP.4} \\
&A^e_{q}=\begin{bmatrix}A_{q} & -B^u_{q} W^{u^{-1}}_{q}{B^u_{q}}^T & B^c_{q}+B^u_{q}\bar{u}_{q} \\-W^x_{q}&-A^T_{q}& W^x_{q}\bar{x}_{q}\\0&0&0\end{bmatrix} \label{E:LP.3} 
\end{align}
\end{subequations}

The above differential equation is solved as 
\begin{subequations}
\label{E:LP}
\begin{align}
&\begin{bmatrix}x_i^-\\\lambda_i^-\\1\end{bmatrix}=\Psi_{q_i}(t^s_i-t^s_{i-1})\begin{bmatrix}x_{i-1}^+\\\lambda_{i-1}^+\\1\end{bmatrix} \quad \forall \ i\in[1..n] \label{E:LP.1}\\
&\Psi_{q}(\alpha) = \text{exp}\left(A^e_{q}\alpha\right) \quad \forall q\in Q, \alpha\in\mathbb{R}\,. \label{E:LP.2} 
\end{align}
\end{subequations}

Also, equation (\ref{E:GB}) is written as 
\begin{align}
&M^x_{q_i,\sigma_i,q_{i+1}}x_i^- + M^c_{q_i,\sigma_i,q_{i+1}} = 0 \quad \forall \ i\in[1..n-1]. \label{E:LPM}
\end{align}

The equations (\ref{E:LP}) and (\ref{E:LPM}) together with $x(t^s_0)=x_0^+=x_{ic}$, (\ref{E:2.2}), (\ref{E:PMPj.1}), (\ref{E:PMPf.1}) with the special forms of the functions in (\ref{E:AHS}) and (\ref{E:QJ}), constitute a system of linear equations in terms of the set of unknowns in $\mathcal{Y}_a$ defined as
\begin{align}
\mathcal{Y}_a =& \big(x_1^-,\cdots,x_n^-,\lambda_1^-,\cdots,\lambda_n^-, \nonumber\\ 
&x_0^+,\cdots,x_{n-1}^+,\lambda_0^+,\cdots,\lambda_{n-1}^+,\alpha_1,\cdots,\alpha_{n-1}). \label{E:FS.2} 
\end{align}

The mentioned system of linear equations can be solved by a matrix inversion. The closed form solution can be represented in terms of $t^s_i$ for $i\in[1..n-1]$ as 
\begin{align}
\mathcal{Y}_a =& F_s^a\big(t^s_1,\cdots,t^s_{n-1}\big) \label{E:FS.1} 
\end{align}
with some $F_s^a:\mathbb{R}^{n-1}\to\mathbb{R}^{4n n_x+n-1}$.

Considering that $u_i^+$ and $u_i^-$ for $i\in[1..n-1]$ are obtained from $\lambda_i^+$ and $\lambda_i^-$ according to (\ref{E:US}), the equation (\ref{E:PMPj.2}) for $i\in[1..n-1]$ can be represented as 
\begin{subequations}
\label{E:FH}
\begin{align}
&F_h^a(\mathcal{Y}_a) = 0, \label{E:FH.1} \\
&\big[F_h^a(\mathcal{Y}_a)\big]_i = H_{q_i}(x_i^-,\lambda_i^-,u_i^-) - H_{q_{i+1}}(x_i^+,\lambda_i^+,u_i^+).  \label{E:FH.2}
\end{align}
\end{subequations}

Replacing $\mathcal{Y}_a$ in (\ref{E:FH.1}) from (\ref{E:FS.1}), we arrive at the set of algebraic equations  
\begin{equation}
F_t^a\big(t^s_1,\cdots,t^s_{n-1}\big)=0 \label{E:FTa}
\end{equation}
with $F_t^a=F_h^a\circ F_s^a$ that can be solved for $t^s_i$, $i\in[1..n-1]$.

Then, $\mathcal{Y}_a$ is obtained from (\ref{E:FS.1}) which allows for computing the remaining elements of the optimal execution. 

\subsection{Calculating the discrete elements}
\label{SS:DC}
In order to apply the indirect MPC to a hybrid system, the Problem~\ref{PB:1} is solved by an algorithm in this part which iterates on the discrete sequences. It solves a number of subproblems either in the form of the Problem~\ref{PB:2} in the previous part or the Problem~\ref{PB:3} defined in the following. 

\begin{problem}
\label{PB:3}
Given a hybrid system $\mathcal{H} = (Q$, $\Sigma$, $\Theta$, $f$, $D$, $G$, $R)$ according to the Definition~\ref{D:HS}, a cost functional $J_m$ as in (\ref{E:JC.1}), initial continuous state $x_{ic}\in\mathbb{R}^{n_x}$, and sequences $q\in Q^{[1..n]}$ and $\sigma\in\Sigma^{[1..n-1]}$ with $n\in\mathbb{Z}_{>0}$, find an execution $E\in$ $\mathcal{E(H}$, $x_{ic})$ with discrete state sequence $q$ and discrete input sequence $\sigma$ which minimizes $J_m$ while satisfying 
\begin{equation}
\label{E:TP3}
t^s_n = t^s_{n-1}. 
\end{equation}
\end{problem}

The above problem is different from the Problem~\ref{PB:2} in that the terminal cost is eliminated and the constraint $t^s_n = T_h$ is replaced with (\ref{E:TP3}). 
The solution of Problem~\ref{PB:3} can be derived from the more general results such as \cite{R:HO4, R:HO3} which requires a considerable space. Another approach is to derive the solution directly from the Proposition~\ref{TH:1} for sufficiently large value of $T_h$ as in the following.

\begin{proposition}
\label{TH:2}
Given a hybrid system $\mathcal{H} = (Q$, $\Sigma$, $\Theta$, $f$, $D$, $G$, $R)$ which satisfies the Assumption~\ref{A:1}, if an execution of $\mathcal{H}$ denoted by $E = (t^s$, $q$, $\sigma$, $x$, $u)$ solves the Problem~\ref{PB:3} for the sequences $q$ and $\sigma$, then there exist $\alpha_i\in\mathbb{R}$ for $i\in[1..n-1]$ with $n=|q|$ and $\lambda:[t^s_0,t^s_n]\to\mathbb{R}^{n_x}$ such that the set of equations (\ref{E:PMPc}) and (\ref{E:PMPj}) for $i\in[1..n-1]$ are satisfied together with
\begin{align}
\lambda_{n-1}^+ &= 0. \label{E:LT2} 
\end{align}
\end{proposition}

\begin{proof}
First, we choose a set $\bar{Q}$ such that $Q\cap\bar{Q}=\emptyset$ and there exist a one to one mapping $\eta:Q\to\bar{Q}$. Then, the hybrid system $\mathcal{H}$ is extended to $\mathcal{H}_e = (Q_e,\Sigma,f,D,G,R)$ with $Q_e=Q\cup\bar{Q}$. The functions $f$, $D$, $G$, $R$, and $h$ are extended such that they assign the same values to $q\in Q$ and $\eta(q)$ in each of their arguments. We also extend $l$ and $h^e$ as $l_{\bar{q}}(\cdot,\cdot) = 0$ and $h^e_{\bar{q}}(\cdot) = 0$ for every $\bar{q}\in\bar{Q}$. 

It is assumed that $T_h$ is larger than the final time of $E$. We denote by $\mathcal{F}$ the subset of executions in $\mathcal{E(H},x_{ic})$ for which the discrete state sequence is $q$, the discrete input sequence is $\sigma$, the final time is less than $T_h$, and (\ref{E:TP3}) is satisfied. Also, we denote by $\mathcal{F}_e$ the set of executions in $\mathcal{E(H}_e,x_{ic},T_h)$ for which the discrete state sequence is $q^e = q_{1..n-1}\,\eta(q_n)$ and the discrete input sequence is $\sigma$.
Then, a mapping $\pi:\mathcal{F}_e\to\mathcal{F}$ can be defined which assigns $(t^s,q,\sigma,x|_{[t^s_0,t^s_n]},u|_{[t^s_0,t^s_n]})\in\mathcal{F}$ to $(t^{se},q^e,\sigma,x,u)\in\mathcal{F}_e$ with $t^{se}_{1..n-1} = t^s_{1..n-1}$. This mapping is surjuctive, because one can construct an element of $\mathcal{F}_e$ given an element of $\mathcal{F}$ by arbitrarily selecting $u$ over $(t^s_{n-1},t^{se}_n]$ (considering that $t^{se}_n\geq t^{se}_{n-1}=t^s_{n-1}$). Hence, the execution $E$ which solves the Problem~\ref{PB:3} for $\mathcal{H}$ with the discrete sequences $q$ and $\sigma$ can be represented as $E=\pi(E_e)$ for some $E_e\in\mathcal{F}_e$. 

It can be easily verified that $J_m(\pi(E'_e)) = J(E'_e)$ for every $E'_e\in\mathcal{F}_e$ by the construction of $\mathcal{H}_e$ and $E_e$. Therefore, $E_e$ must solve the Problem~\ref{PB:2} for $\mathcal{H}_e$ with the discrete sequences $q^e$, $\sigma$. Because, if there exist $E'_e\in\mathcal{F}_e$ which gives $J(E'_e)<J(E_e)$, then we have $J_m(\pi(E'_e))<J_m(\pi(E_e))$ which contradicts with the assumption that $\pi(E_e)=E$ solves the Problem~\ref{PB:3}.

Applying the Proposition~\ref{TH:1} to $E_e$, it is concluded that (\ref{E:PMPc}) for $i\in[1..n]$, (\ref{E:PMPj}) for $i\in[1..n-1]$, and (\ref{E:PMPf.1}) hold for the extended system $\mathcal{H}_e$. By the construction of $\mathcal{H}_e$, the equations (\ref{E:PMPc}) for $i\in[1..n-1]$ and (\ref{E:PMPj}) for $i\in[1..n-1]$ are in terms of the elements of the original system $\mathcal{H}$. This proves the result except for the equation (\ref{E:LT2}). The equations (\ref{E:PMPc.1}) for $i=n$ and (\ref{E:PMPf.1}) are written as $\dot{\lambda} = - D^T_x f(\bar{q}_n,x,u)\lambda$ for $t\in[t^s_{n-1},t^{se}_{n})$ and $\lambda(t^{se}_{n})=0$ which can be solved as $\lambda(t)=0$ over $t\in[t^s_{n-1},t^{se}_{n}]$ to obtain (\ref{E:LT2}).
\end{proof}

Solution of the Problem~\ref{PB:3} in the case of affine hybrid systems is obtained by modifying the set of linear equations in the Subsection~\ref{SS:AHS} that must be solved to obtain (\ref{E:FS.1}). The modification includes removing (\ref{E:LP.1}) for $i=n$ from the set of equations, and correspondingly removing $x_n^-$ and $\lambda_n^-$ from the set of unknowns $\mathcal{Y}_a$ in (\ref{E:FS.2}) to obtain a new set of unknowns $\mathcal{Y}_b$. Also, the equations $t^s_n = T_h$ and (\ref{E:PMPf.1}) are replaced with (\ref{E:TP3}) and (\ref{E:LT2}). The modified version of (\ref{E:FS.1}) for Problem~\ref{PB:3} is written as 
\begin{align}
\mathcal{Y}_b =& F_s^b\big(t^s_1,\cdots,t^s_{n-1}\big). \label{E:FS.1m} 
\end{align}

The function $F_h^a$ in (\ref{E:FH}) should be also modified to a new function $F_h^b$ which accepts $\mathcal{Y}_b$ as its argument and has the same definition as in (\ref{E:FH.2}). Then, the Equation (\ref{E:FTa}) becomes 
\begin{equation}
F_t^b\big(t^s_1,\cdots,t^s_{n-1}\big)=0 \label{E:FTb}
\end{equation}
in which $F_t^b=F_h^b\circ F_s^b$. 

Applications of the propositions \ref{TH:1} and \ref{TH:2} for solving the problems \ref{PB:2} and \ref{PB:3} in the case of affine hybrid systems, are respectively represented as the subroutines $\mathtt{JPMPa}$ and $\mathtt{JPMPb}$ in the following. 
Efficient numerical procedures for calculation of $\mathcal{Y}_a$, $\mathcal{Y}_b$, $J$, and $J_m$ that are the basic operations in $\mathtt{JPMPa}$ and $\mathtt{JPMPb}$ are proposed in the appendices.
Using these functions, the hybrid MPC calculations are accomplished according to the Algorithm~\ref{ALG:1} in the following which is based on the branch and bound method. 

\begin{algorithm2e}[!btp]
\SetAlgoHangIndent{1em}
\SetKwFunction{JPMPa}{JPMPa}
\SetKwFunction{JPMPb}{JPMPb}
\SetKwProg{Fn}{function}{}{end}
\Fn {$\JPMPa(x_{ic},\sigma,q)$}{
Solve equation (\ref{E:FTa}) for $t^s_{1..n-1}$ with $t^s_0=0$, $t^s_n=T_h$, and $x_0^+=x_{ic}$. \;
\uIf{the solution $t^s_{1..n-1}$ exists}{Calculate, $\mathcal{Y}_a$ from (\ref{E:FS.1}) and $J$ from (\ref{E:JC.2}).}
\Else{Set $u_0^+$ to 0 and $J$ to $\infty$.}
\KwRet{$(u_0^+,J)$} \;
} 
\vspace{0.4cm}
\setcounter{AlgoLine}{0}
\Fn {$\JPMPb(x_{ic},\sigma,q)$}{
Solve equation (\ref{E:FTb}) for $t^s_{1..n-1}$ with $t^s_0=0$ and $x_0^+=x_{ic}$. \; 
\uIf{the solution $t^s_{1..n-1}$ exists}{Calculate, $\mathcal{Y}_b$ from (\ref{E:FS.1m}) and $J_m$ from (\ref{E:JC.1}).}
\Else{Set $u_0^+$ to 0 and $J_m$ to $\infty$.}
\KwRet{$(u_0^+,J_m)$} \;
} 
\end{algorithm2e}

The algorithm gets the current discrete and continuous states $q_{ic}\in Q$,  $x_{ic}\in\mathbb{R}^{n_x}$ and returns the discrete and continuous inputs $\sigma_{ap}\in\Sigma$, $u_{ap}\in\mathbb{R}^{n_u}$ that should be applied to the hybrid plant. 
Each element $s=(\nu,q,\sigma,u_0,J)$ of the set $\mathcal{S}$ contains an assessed pair of discrete sequences $q$, $\sigma$. The value of $J$, is the associated optimal cost in Problem~\ref{PB:2} or Problem~\ref{PB:3} if $\nu=1$ or $\nu=0$ respectively. As shown in the Lemma~\ref{LM:2} in the next section, if $\nu=0$, then $J$ is a lower bound of the cost value for all executions in $\mathcal{E(H}$, $x_{ic}$, $T_h)$ whose discrete sequences have $q$ and $\sigma$ as their prefixes. Therefore, the element $\hat{s}=(\hat{\nu},\hat{\sigma},\hat{q},\hat{u}_0,\hat{J})\in\mathcal{S}$ which has the minimum value of $J$ among the elements of $\mathcal{S}$ solves the Problem~\ref{PB:1} if $\hat{\nu}=1$. Otherwise, if $\hat{\nu}=0$, the algorithm branches $\hat{s}$ until the optimal solution is found.
After completion of the algorithm, $\hat{q}$ and $\hat{\sigma}$ are the optimal sequences of discrete state and discrete input respectively.
More details on the operation of the algorithm together with the proof of its correctness are provided in the next section.

\IncMargin{1em}
\begin{algorithm2e}[t]
\SetAlgoHangIndent{2em}
\SetKwFunction{JPMPa}{JPMPa}
\SetKwFunction{JPMPb}{JPMPb}
\SetKwInOut{Input}{input}
\SetKwInOut{Output}{output}
\Input{$q_{ic}$, $x_{ic}$}
\Output{$\sigma_{ap}$, $u_{ap}$}
\vspace{3pt}
$(\hat{\nu},\hat{\sigma},\hat{q},\hat{u}_0,\hat{J}) \gets (0,\{\},q_{ic},0,0)$ \;
$\mathcal{S} \gets \{(\hat{\nu},\hat{\sigma},\hat{q},\hat{u}_0,\hat{J})\}$ \;
\While{\label{ALG:1.2} $\hat{\nu}=0$}{
$(u^c_0,J_c) \gets \JPMPa(x_{ic},\hat{\sigma},\hat{q})$ \label{ALG:1.8}\;
$\mathcal{S} \gets \mathcal{S} \cup \{(1,\hat{\sigma},\hat{q},u^c_0,J_c)\}$  \label{ALG:1.4}\;
\For{$(q,\sigma,q')\in\Theta$ \textnormal{such that} $q=\hat{q}_{|\hat{q}|}$ \label{ALG:1.5}}{
$(u^c_0,J_c) \gets \JPMPb(x_{ic},\hat{\sigma}\sigma',\hat{q}q')$ \label{ALG:1.9}\;
$\mathcal{S} \gets \mathcal{S} \cup \{(0,\hat{\sigma}\sigma',\hat{q}q',u^c_0,J_c)\}$ \label{ALG:1.1}\;
}
$\mathcal{S} \gets \mathcal{S}\setminus \{(\hat{\nu},\hat{\sigma},\hat{q},\hat{u}_0,\hat{J})\}$ \label{ALG:1.6}\;
$(\hat{\nu}, \hat{\sigma}, \hat{q}, \hat{u}_0, \hat{J}) \gets \underset{\text{such that} J\leq J' \text{ for all} (\nu',\sigma',q',u'_0,J')\in \mathcal{S}}{(\nu, \sigma, q, u_0, J) \in\mathcal{S}\hfill}$ 
\label{ALG:1.3}\;
}
$\sigma_{ap}\gets \hat{\sigma}_1$ \;
$u_{ap}\gets \hat{u}_0$ \;
\vspace{3pt}
\caption{Indirect MPC Algorithm}
\label{ALG:1}
\end{algorithm2e}
\DecMargin{1em}

To extend Algorithm~\ref{ALG:1} for general hybrid systems, the subroutine $\mathtt{JPMPa}$ ($\mathtt{JPMPb}$) should be modified such that it solves the differential-algebraic system of equations given by (\ref{E:2.1}), (\ref{E:2.2}), (\ref{E:GB}) and the conditions in Proposition~\ref{TH:1} (\ref{TH:2}).

\section{Correctness and finiteness of the indirect MPC algorithm}
\label{S:ALP}
First, we define a few additional notations and provide some useful lemmas. For every $\bar{s}=(\bar{\nu}$, $\bar{\sigma}$, $\bar{q}$, $\bar{u}_0$, $\bar{J})\in\mathcal{S}$ in the Algorithm~\ref{ALG:1}, we denote $\bar{\nu}$, $\bar{\sigma}$, and $\bar{q}$ as the $\nu$-component, $\sigma$-component, and $q$-component of $\bar{s}$ respectively. The operations within the while loop at line~\ref{ALG:1.2} of the algorithm constitute an iteration of the algorithm. To indicate the value of a variable at an iteration, the iteration number is added as superscript. 
For example, at line~\ref{ALG:1.2} of the $i$th iteration, the value of the set $\mathcal{S}$ is denoted by $\mathcal{S}^i$ and the value of $(\hat{\nu}$, $\hat{\sigma}$, $\hat{q}$, $\hat{u}_0$, $\hat{J})$ is denoted by $(\hat{\nu}^i$, $\hat{\sigma}^i$, $\hat{q}^i$, $\hat{u}^i_0$, $\hat{J}^i)$. 
Given $\sigma\in\Sigma^{[1..n-1]}$ and $q\in Q^{[1..n]}$ for some $n\in\mathbb{Z}_{\geq 0}$, a function $E_{opt}$ is defined such that the executions obtained in the subroutines $\mathtt{JPMPa}$ and $\mathtt{JPMPb}$ for $x_{ic}\in\mathbb{R}^{n_x}$ are given by $E_{opt}(1,\sigma,q,x_{ic})$ and $E_{opt}(0,\sigma,q,x_{ic})$ respectively. 
For every $\bar{s}=(\bar{\nu}$, $\bar{\sigma}$, $\bar{q}$, $\bar{u}_0$, $\bar{J})\in\mathcal{S}$ at every iteration, if $\bar{\nu}=1$, then $\bar{s}$ is added to $\mathcal{S}$ at line~\ref{ALG:1.4} of the algorithm and $\bar{J}$ is obtained at line~\ref{ALG:1.8} which implies $\bar{J}=E_{opt}(1,\bar{\sigma},\bar{q},x_{ic})$. Otherwise, $\bar{\nu}=0$ and $\bar{s}$ is added to $\mathcal{S}$ at line~\ref{ALG:1.1} of the algorithm. Then, the element $\bar{J}$ is obtained at line~\ref{ALG:1.9} which implies $\bar{J}=E_{opt}(0,\bar{\sigma},\bar{q},x_{ic})$. Therefore, we can write 

\begin{equation}
\label{E:EO}
\bar{J}=E_{opt}(\bar{\nu},\bar{\sigma},\bar{q},x_{ic}) \qquad \forall\,(\bar{\nu},\bar{\sigma},\bar{q},\bar{u}_0,\bar{J})\in\mathcal{S}
\end{equation}

\begin{lemma}
\label{LM:1}
If the Algorithm~\ref{ALG:1} is applied to a hybrid system $\mathcal{H} = (Q$, $\Sigma$, $\Theta$, $f$, $D$, $G$, $R)$ with the cost functional $J$ in (\ref{E:JC}) and $q_{ic}\in Q$, then for every $i\in\mathbb{Z}_{>0}$, $\bartt{q}\in Q^{[1..n]}$, $\bartt{\sigma}\in\Sigma^{[1..n-1]}$ with $n\in\mathbb{Z}_{>0}$ and $\bartt{q}_1=q_{ic}$, there exist an element $(\bar{\nu}$, $\bar{\sigma}$, $\bar{q}$, $\bar{u}_0$, $\bar{J})\in\mathcal{S}^i$ such that either $\vartheta^i_a$ or $\vartheta^i_b$ is true where

\begin{subequations}
\label{E:3} 
\begin{align}
\vartheta^i_a &= \big[ \bar{\nu}=1 \wedge \bar{\sigma}=\bartt{\sigma} \wedge \bar{q}=\bartt{q} \big], \label{E:3.2} \\
\vartheta^i_b &= \big[ \bar{\nu}=0 \wedge \bar{\sigma}\prec\bartt{\sigma} \wedge \bar{q}\prec\bartt{q} \big]. \label{E:3.3} 
\end{align}
\end{subequations}
\end{lemma}

\begin{proof}
The lemma is proved via induction. For $i=1$ in which case $\mathcal{S}^1$ is set by the first two lines of the algorithm, $\vartheta^1_b$ is true. If $\hat{\nu}^i=0$ at the $i$th iteration, then the algorithm continues to the $(i+1)$th iteration. Assuming that $\vartheta^i_a \vee \vartheta^i_b$ is true for some $(\bar{\nu}$, $\bar{\sigma}$, $\bar{q}$, $\bar{u}_0$, $\bar{J})\in\mathcal{S}^i$, we must have $|\bar{q}|\leq|\bartt{q}|$ due to (\ref{E:3}) and there can be four cases:
\begin{itemize}
\item $[\bar{\nu}=1]$. In this case, every element whose $\nu$-component is one, including $(\bar{\nu},\bar{\sigma}$, $\bar{q}$, $\bar{u}_0$, $\bar{J})$ remains in $\mathcal{S}^{i+1}$ and cannot be the element which is removed at line~\ref{ALG:1.6}. Hence, $\vartheta^{i+1}_a$ is true.
\item $[\bar{\nu}=0 \ \wedge \ \neg(\bar{q}=\hat{q}^i \ \wedge \ \bar{\sigma}=\hat{\sigma}^i)]$. In this case, every element whose $q$ and $\sigma$ components are respectively different from $\hat{q}^i$ and $\hat{\sigma}^i$, including $(\bar{\nu},\bar{\sigma}$, $\bar{q}$, $\bar{u}_0$, $\bar{J})$ remains in $\mathcal{S}^{i+1}$ and cannot be the element which is removed at line~\ref{ALG:1.6}. Hence, $\vartheta^{i+1}_b$ is true.
\item $[\bar{\nu}=0 \ \wedge \ (\bar{q}=\hat{q}^i \ \wedge \ \bar{\sigma}=\hat{\sigma}^i) \ \wedge \ |\bar{q}|=|\bartt{q}|]$. In this case, $(1,\bar{\sigma}$, $\bar{q}$, $\bar{u}_0$, $\bar{J})$ is included in $\mathcal{S}^{i+1}$ at line~\ref{ALG:1.4}. Hence, $\vartheta^{i+1}_a$ is true.
\item $[\bar{\nu}=0 \ \wedge \ (\bar{q}=\hat{q}^i \ \wedge \ \bar{\sigma}=\hat{\sigma}^i) \ \wedge \ |\bar{q}|<|\bartt{q}|]$. In this case, 
$(0,\bar{\sigma}\bartt{\sigma}_{|\bar{q}|}$, $\bar{q}\bartt{q}_{|\bar{q}|+1}$, $u^c_0$, $J_c)$ for some $u^c_0$ and $J_c$ is included in $\mathcal{S}^{i+1}$ at line~\ref{ALG:1.1} and $\vartheta^{i+1}_b$ is true.
\end{itemize}
Therefore, $\vartheta^i_a \vee \vartheta^i_b$ remains true at the $(i+1)$th iteration in all of the cases. 
\end{proof}

\begin{lemma}
\label{LM:2}
Considering a hybrid system $\mathcal{H}$, initial continuous state $x_{ic}\in\mathbb{R}^{n_x}$, and a cost functional $J$ as in (\ref{E:JC}), for every $\bartt{E} = (\bartt{{t^s}}$, $\bartt{q}$, $\bartt{\sigma}$, $\bartt{x}$, $\bartt{u})\in\mathcal{E(H}$, $x_{ic})$, $\bar{q}\in\Sigma^{[1..n]}$, and $\bar{\sigma}\in Q^{[1..n-1]}$ with $n\in\mathbb{Z}_{>0}$, if we have $\bar{q} \prec \bartt{q}$ and $\bar{\sigma} \prec \bartt{\sigma}$, then $J(\bartt{E})\geq J_m(E_{opt}(0,\bar{\sigma},\bar{q},x_{ic}))$. 
\end{lemma}

\begin{proof}
If the conditions $\bar{q} \prec \bartt{q}$ and $\bar{\sigma} \prec \bartt{\sigma}$ hold, then $\bartt{E}$ can be trimmed into an execution $E^m = (t^{sm}$, $q^m$, $\sigma^m$, $x^m$, $u^m)\in\mathcal{E(H},x_{ic})$ which is a prefix of $\bartt{E}$ and satisfies $q^m = \bar{q}$, $\sigma^m = \bar{\sigma}$, and $t^{sm}_n=t^{sm}_{n-1}$. Since $E_m$ is a prefix of $\bartt{E}$, one can write $J_m(E_m)\leq J(\bartt{E})$ according to (\ref{E:JC}). On the other hand, we have $J_m(E_{opt}(0,\bar{\sigma},\bar{q},x_{ic})) \leq J_m(E_m)$, since by the definition, $E_{opt}(0,\bar{\sigma},\bar{q},x_{ic})$ solves the Problem~\ref{PB:3}. The combination of these two inequalities proves the lemma.
\end{proof}

\begin{lemma}
\label{LM:3}
If the Algorithm~\ref{ALG:1} is applied to a hybrid system $\mathcal{H}$ with the cost functional $J$ in (\ref{E:JC}), $x_{ic}\in\mathbb{R}^{n_x}$, and $q_{ic}\in Q$, then for every $i\in\mathbb{Z}_{>0}$ and for every execution $\bartt{E} = (\bartt{{t^s}}, \bartt{q}, \bartt{\sigma}, \bartt{x}, \bartt{u})\in\mathcal{E(H}$, $x_{ic}$, $T_h)$ which satisfies $\bartt{q}_1=q_{ic}$, we have that $J(\bartt{E})\geq \hat{J}^i$.
\end{lemma}

\begin{proof}
The Lemma~\ref{LM:1} implies that for every $i\in\mathbb{Z}_{>0}$ there exist $(\bar{\nu}$, $\bar{\sigma}$, $\bar{q}$, $\bar{u}$, $\bar{J})\in\mathcal{S}^i$ such that either $\vartheta^i_a$ or $\vartheta^i_b$ in (\ref{E:3}) is true for the elements $\bartt{q}$ and $\bartt{\sigma}$ of the execution $\bartt{E}$. First, it is shown that we have $J(\bartt{E})\geq \bar{J}$.
If $\vartheta^i_a$ is true, then $\bar{\nu}=1$, $\bartt{q}=\bar{q}$, and $\bartt{\sigma}=\bar{\sigma}$. Then, $E_{opt}(\bar{\nu}$, $\bar{\sigma}$, $\bar{q}$, $x_{ic})$ solves the Problem~\ref{PB:2} and we have $J(\bartt{E})\geq J(E_{opt}(\bar{\nu}$, $\bar{\sigma}$, $\bar{q}$, $x_{ic}))=\bar{J}$ according to (\ref{E:EO}). 
On the other hand, if $\vartheta^i_b$ is true, then $\bar{\nu}=0$, $\bartt{q}\prec\bar{q}$, and $\bartt{\sigma}\prec\bar{\sigma}$. According to the Lemma~\ref{LM:2}, we have $J(\bartt{E})\geq J_m(E_{opt}(\bar{\nu}$, $\bar{\sigma}$, $\bar{q}$, $x_{ic}))=\bar{J}$. 
Hence, $J(\bartt{E})\geq \bar{J}$ holds in every case.
The operation at line~\ref{ALG:1.3} of the algorithm requires that $\bar{J}\geq\hat{J}^i$ which together with $J(\bartt{E})\geq \bar{J}$ gives $J(\bartt{E})\geq \hat{J}^i$.
\end{proof}

\subsection{Correctness of the algorithm}

The following result establishes the correctness of the Algorithm~\ref{ALG:1}.

\begin{theorem}
\label{TH:3}
If the Algorithm~\ref{ALG:1} is applied to a hybrid system $\mathcal{H}$ with the cost functional $J$ in (\ref{E:JC}), the initial continuous state $x_{ic}\in\mathbb{R}^{n_x}$, and the initial discrete state $q_{ic}\in Q$, then after the termination of the algorithm, the execution $\hat{E}=E_{opt}(1$, $\hat{\sigma}$, $\hat{q}$, $x_{ic})$ solves the Problem~\ref{PB:1}.
\end{theorem}

\begin{proof} 
Denoting the total number of iterations by $\ell$, the condition at line~\ref{ALG:1.2} requires that $\hat{\nu}^\ell = 1$ at the final iteration. 
Hence, we can write $\hat{E}=E_{opt}(\hat{\nu}^\ell$, $\hat{\sigma}^\ell$, $\hat{q}^\ell$, $x_{ic})$. Also, we have $\hat{J}^\ell=J(E_{opt}(\hat{\nu}^\ell$, $\hat{\sigma}^\ell$, $\hat{q}^\ell$, $x_{ic}))$ due to (\ref{E:EO}) which then gives $\hat{J}^\ell=J(\hat{E})$.
Application of the Lemma~\ref{LM:3} with $i=\ell$ results in $J(\bartt{E})\geq$ $\hat{J}^\ell=$ $J(\hat{E})$ for every $\bartt{E} = (\bartt{{t^s}}, \bartt{q}, \bartt{\sigma}, \bartt{x}, \bartt{u})\in\mathcal{E(H}$, $x_{ic}$, $T_h)$ with $\bartt{q}_1=q_{ic}$. Therefore, the execution $\hat{E}$ solves the Problem~\ref{PB:1}. 
\end{proof}

\subsection{Finiteness of the algorithm}

In the general case, there is no upper bound on the length of the optimal state sequence for the execution which solves the Problem~\ref{PB:1}. However, for applying the indirect hybrid MPC algorithm in practice, it is important to ensure that the Algorithm~\ref{ALG:1} terminates in a finite number of steps. In this part, two solutions are proposed for managing the number of iterations of the algorithm. 

Considering a hybrid system $\mathcal{H} = (Q$, $\Sigma$, $\Theta$, $f$, $D$, $G$, $R)$, we enumerate the elements of $Q$ as $Q=\{q^e_1$, $q^e_2$, $\cdots$, $q^e_{|Q|}\}$. Then, the set of jumps $\Theta$ is converted to a matrix $\Theta_a$ defined as 
\begin{subequations}
\begin{align*}
[\Theta_a]_{i,j} &= |\varsigma(q^e_i,q^e_j)| & i,j\in[1..|Q|],\\
\varsigma(q,q') &= \{\sigma\in\Sigma \, | \, (q,\sigma,q')\in\Theta\} & \forall\,q,q'\in Q.
\end{align*}
\end{subequations}

In fact, $\Theta_a$ is the adjacency matrix of a directed multigraph $\mathcal{G(H)}$ with the set of vertices $Q$ and the set of edges given by $\Theta$. There is an edge from $q^e_i$ to $q^e_j$ for every $(q^e_i,\sigma,q^e_j)\in\Theta$. It is known that the number of the directed paths of length $k\in\mathbb{Z}_{>0}$ from $q^e_i$ to $q^e_j$ in $\mathcal{G(H)}$ is given by $[\Theta_a^k]_{i,j}$ \cite{R:GT1}.
For every execution $(t^s,q,\sigma,x,u)\in\mathcal{E(H)}$ with $|q|=n\in\mathbb{Z}_{>0}$, there exist a directed path $\theta\in\Theta^{[1..n-1]}$ in $\mathcal{G(H)}$ such that $\theta_i=(q_i,\sigma_i,q_{i+1})$ for $i\in[1..n-1]$. Therefore, the number of all possibilities for the discrete sequences $q$ and $\sigma$ of an execution with $|q|\leq m\in\mathbb{Z}_{>0}$ and $q_1=q_{ic}$ can be computed as 
\begin{align}
\label{E:NS}
\mathcal{N}(m,q_{ic}) &= \sum_{i=0}^{m-1} \sum_{j=1}^{|Q|} [\Theta_a^{i}]_{q_{ic},j} 
= [(\textstyle{\sum_{i=0}^{m-1}} \Theta_a^{i})\,\mathbf{1}]_{q_{ic}} . 
\end{align}

\noindent
in which all of the elements of $\mathbf{1}\in\mathbb{R}^{|Q|}$ are equal to one.

By defining $n_a$ as (\ref{E:NSB2}) in the following, one can write the element-wise inequality $\Theta_a \mathbf{1}\leq n_a \mathbf{1}$. This inequality can be applied repeatedly to obtain $\mathcal{N}(m,q_{ic})\leq\sum_{i=0}^{m-1}n_a^i$ which can be represented as in (\ref{E:NSB1}).

\begin{subequations}
\begin{align}
\mathcal{N}(m,q_{ic}) &\leq \begin{cases}\frac{n_a^{m}-1}{n_a-1} & n_a>1 \\ m & n_a=1\end{cases} \label{E:NSB1} \\
n_a&=\max_{i\in[1..|Q|]}[\Theta_a \mathbf{1}]_i \label{E:NSB2}
\end{align}
\end{subequations}

The solutions for assuring the finiteness of the Algorithm~\ref{ALG:1} are based on the following lemma.

\begin{lemma}
\label{LM:4}
Considering a hybrid system $\mathcal{H}$ and initial discrete state $q_{ic}$, the number of iterations of the Algorithm~\ref{ALG:1} before reaching the condition $|\hat{q}|>m$ for an integer $m\in\mathbb{Z}_{>0}$ is less than $\mathcal{N}(m,q_{ic})$ given in (\ref{E:NS}).
\end{lemma}

\begin{proof} 
If the condition $|\hat{q}|>m$ is not reached until the $i$th iteration, then we must have $|\bar{q}|\leq m$ for every $(\bar{\nu}$, $\bar{\sigma}$, $\bar{q}$, $\bar{u}$, $\bar{J})\in\mathcal{S}^i$.
Because, if $|\bar{q}|>m$, then it is necessary that $|\hat{q}^j|=m$ for some $j<i$ such that the algorithm can generate elements in $\mathcal{S}^{j+1}$ with $q$-component longer than $m$ at line~\ref{ALG:1.1} in the $j$th iterations. However, this contradicts with the assumption that the condition $|\hat{q}|>m$ is not reached until the $i$th iteration. 

The number of all possible combinations of $q\in Q^{[1..n]}$ and $\sigma\in\Sigma^{[1..n-1]}$ with $q_1=q_{ic}$ and $n\leq m$ is calculated as $\mathcal{N}(m,q_{ic})$. 
Each of these possibilities may appear as a pair of $\hat{q}^i$ and $\hat{\sigma}^i$ with $\hat{\nu}^i=0$ for some $i$ to initiate an iteration. Therefore, in the worst case, the maximum number of iterations would be $\mathcal{N}(m,q_{ic})$.
\end{proof} 

The first solution for keeping the number of iterations finite, is to terminate the algorithm whenever $|\hat{q}|$ exceeds a prescribed value according to the following corollary which is a direct consequence of the Lemma~\ref{LM:4}.

\begin{corollary}
\label{TH:4}
Considering a hybrid system $\mathcal{H}$, if the condition $\hat{\nu}=0$ at line~\ref{ALG:1.2} of the algorithm is modified to $\hat{\nu}=0 \wedge |\hat{q}|\leq m$, then the number of iterations of the Algorithm~\ref{ALG:1} before completion will be less than $\mathcal{N}(m,q_{ic})$ given in (\ref{E:NS}).
\end{corollary}

However, the limitation on the length of the $\hat{q}$ in the above corollary may result in a suboptimal control at each step of the MPC algorithm (if the optimal length of the discrete state sequence becomes greater than $m$). 
The second solution for ensuring the finiteness of algorithm is to select the function $h$ in (\ref{E:JC.1}) such that it is lower bounded according to the following proposition. 

\begin{theorem}
\label{TH:5}
Considering a hybrid system $\mathcal{H}$ and the execution $E$ which solves the Problem~\ref{PB:1} for $\mathcal{H}$ with the cost functional $J$ in (\ref{E:JC}), if the function $h$ is lower bounded by $h_{min}$, then the number of iterations of the Algorithm~\ref{ALG:1} will be less than $\mathcal{N}(m,q_{ic})$ given in (\ref{E:NS}) with $m= 1+\lfloor J(E)/h_{min}\rfloor$.
\end{theorem}

\begin{proof} 
Applying the Lemma~\ref{LM:3} with $\bartt{E}=E$ we have $J(E)\geq\hat{J}^i$ for every $i$. Since $h$ is lower bounded by $h_{min}$, we have $\hat{J}^i\geq h_{min}(|\hat{q}^i|-1)$ according to the definition of the cost functional in (\ref{E:JC}). This in combination with $J(E)\geq\hat{J}^i$ results in $|\hat{q}^i|\leq 1+\lfloor J(E)/h_{min}\rfloor$. Then, application of the Lemma~\ref{LM:4} completes the proof.
\end{proof} 

Since $J(E)$ is not known a priori, one can replace it with $J(E_a)$ for an arbitrary execution $E_a\in\mathcal{E(H)}$ to obtain a larger upper bound for the number of iterations. Because, we always have $J(E_a)\geq J(E)$ and $\mathcal{N}$ in (\ref{E:NS}) is non-decreasing with respect to $m$. For example, one can select some discrete sequences $q$, $\sigma$ (the most simple choice is $q=q_{ic}$, $\sigma=\{\}$) and compute $J(E_a)$ as $J(E_a)=\mathtt{\JPMPa}(x_{ic},\sigma,q)$.

\section{Remarks on Extending the Results} 
\label{S:RX}
In this section, brief comments are provided for some important aspects that cannot be fully addressed in this paper due to the limited space.

\subsection{Piecewise affine hybrid systems}
\label{SS:GPW}
In some practical hybrid systems, the function elements $f$, $g$, and $R$ may be affine either naturally or approximately. For example, guard conditions are in the form of threshold values for state or output variables in most of the cases. 
Otherwise, it should be possible to approximate these function elements with piecewise affine (PWA) functions. PWA functions are treated vary naturally in the framework of hybrid systems \cite{R:HYS}. 
For example, the surface $g_{q,\sigma,q'}(x)=0$ in (\ref{E:GB}) can be approximated as $x\in\Omega_j \Rightarrow M_{q,\sigma,q'}^{x,j} x + M_{q,\sigma,q'}^{c,j} = 0$ such that  $\mathbb{R}^{n_x}$ is partitioned by $\Omega_j$ for $1\le j\le m$. 
Then, new discrete inputs $\sigma_i$ for $1\le j\le m$ are defined such that $(q,\sigma,q')\in\Theta$ is decomposed to multiple transitions $(q,\sigma_j,q')$ for $1\le j\le m$ with $g_{q,\sigma_j,q'}(x) = M_{q,\sigma,q'}^{x,j} x+M_{q,\sigma,q'}^{c,j}$. To consider the partitions, a small modification should be made in the functions $\mathtt{\JPMPa}$ and $\mathtt{\JPMPb}$ in the Algorithm 1 on page 5 such that the obtained solution is acceptable if the jump states ($x_i^-$ for $i\in{1..n-1}$) belong to the corresponding partitions. 

\subsection{Constraints}
A useful feature of the MPC method is the possibility of imposing constraints on the system variables within the prediction horizon. It is possible to retain this feature in the indirect MPC using the versions of the HMP for constrained optimal control \cite{R:HO5}. 
An intermediate situation is to have the element-wise inequality constraints (\ref{E:IC.1}) in the following for every $(q,\sigma,q')\in\Theta$
that are imposed at $t^s_i, i\in[1..n]$. 
According to some versions of the HMP (e.g. \cite{R:HO4, R:HO3}), if the Problem~\ref{PB:2} additionally requires (\ref{E:IC.1}), then the Proposition~\ref{TH:1} is modified such that a term $\beta_i^T D_x^T \phi_{q_i,\sigma_i,q_{i+1}}(x_i^-)$ with $\beta_i\in\mathbb{R}^{n_{q_i,\sigma_i,q_{i+1}}}$ satisfying 
(\ref{E:IC.3}) is added to the right hand side of (\ref{E:PMPj.1}), the same term for $i=n$ is added to the right hand side of (\ref{E:PMPf.1}), and the Equation (\ref{E:IC.2}) holds for $j\in[1..n_{q_i,\sigma_i,q_{i+1}}]$.
\begin{align}
\phi_{q_i,\sigma_i,q_{i+1}}(x_i^-)&\leq 0 \qquad i\in[1..n] \label{E:IC.1}\\
\beta_i&\geq 0 \qquad i\in[1..n] \label{E:IC.3}\\
[\beta_i]_j [\phi_{q_i,\sigma_i,q_{i+1}}(x_i^-)]_j &= 0 \qquad i\in[1..n] \label{E:IC.2}
\end{align}

Input constraints on $u_i^-, i\in[1..n]$, can be handled by replacing \eqref{E:US} with the corresponding Karush-Kuhn-Tucker (KKT) conditions which increases the complexity of calculations. A better idea is to convert the continuous inputs to continuous states by appending integrators at the inputs to treat the input constraints as state constraints in the form of \eqref{E:IC.1}.

Consider a hybrid system with the set of discrete states $Q$, and an execution of it with time sequence $t^s$. To impose constraints at an arbitrary time during the flows within the prediction horizon $t'\in(t^s_0,t^s_n)$ when the discrete state is $q\in Q$, one can virtually add an ineffective jump at $t'$ from $q$ to itself. Then, an additional constraint in the form of (\ref{E:IC.1}) can be imposed at $t'$. For this purpose, an auxiliary state variable $\rho$ with $\dot{\rho}=1$ during flows and $\rho_i^+=\rho_i^-$ at jumps is added to the continuous state $x$ in order to measure time. Then, the condition for the ineffective jump can be represented as $g_{q,\sigma,q}(x)=\rho-t'=0$ for every $q\in Q$, $\sigma\in\Sigma$. Of course, $\rho$ can be used to apply constraints at an arbitrary number of time instants within the prediction horizon. 
It is mentioned that the HMP can be extended to the case in which (\ref{E:GB}) is time dependent such that there is no need to define $\rho$.

\subsection{Stable MPC}
A basic requirement for every control system is stability. There are two means of achieving stability in MPC algorithms \cite{R:HMPC1}: constraint or cost on the final state $x(t+T_h)$. 
There exist results on stability and recursive feasibility of the MPC for discrete-time hybrid systems \cite{R:Lazar2006, R:NMPC2, R:HMPC1}. It is not difficult to modify these results for the MPC formulation in Subsection~\ref{SS:PB}.

\section{Case Study}
\label{S:EXM}
This section presents an application of the proposed method to the supermarket refrigeration system in \cite{R:SU1, R:SU2}, during which comparisons are made with the MPC approach in \cite{R:HMPC4,R:MLD,R:HMPC1} denoted as MLD-MPC. This system is composed of $n_d>1$ display cases and some compressors for circulation of the refrigerant fluid. The set of equations that determine the hybrid dynamics of the system are provided in \cite{R:SU2}. The control inputs are the evaporator inlet valves $valve_i\in\{0,1\}$, $i\in[1..n_d]$ that are discrete and the compressing capacity $comp$ which is continuous. The objective is to control the air temperature in the display cases $T_{a,i}$, $i\in[1..n_d]$ and the suction pressure $P_{suc}$ with minimal control effort. 
A traditional control system is described in \cite{R:SU2} which is composed of $n_d$ hysteresis controllers for adjusting the air temperature in each of the display cases and a PI controller with deadband for regulating the suction pressure. A shortcomings of this controller is the tendency to synchronize the switching times of the inlet valves which causes fluctuations, reduces efficiency, and damages the compressor. Two different MPC solutions are applied to the system in \cite{R:SU1, R:SU3, R:SU4}. The first solution in \cite{R:SU1} applies the hybrid MPC method of \cite{R:MLD} which faces issues when the controller execution period is small. To avoid these issues, the PI controller for the suction pressure is unaltered in \cite{R:SU3, R:SU4} and a nonlinear MPC algorithm is applied to determine only the switching times of the valves. 

\subsection{Applying the Results}
To use the indirect MPC algorithm for controlling the whole refrigeration system, it is first required to define a cost function. 
In order to assign cost to variations of $comp$ (similar to \cite{R:SU1}), a new input is defined as $u = \frac{d}{dt}comp$. Also, another state variable $\rho$ is defined for assigning cost to short switching time intervals such that $\dot{\rho} = -w_1 \rho$ during flows and $\rho^+_i = w_2 \rho^-_i + w_3$ at a switching time instant $t^s_i$ ($w_1,w_2$, and $w_3$ are design parameters). 
There are nonlinearities in equations of the system in \cite{R:SU2}. To approximate them by affine equations, the right hand sides of the equations in the appendix A of \cite{R:SU2} are approximated by constant values evaluated for $P_{suc}=P_{des}$ and the right hand side of equation (6) in \cite{R:SU2} is approximated by a linear function. 
A system with two display cases is considered ($n_d = 2$). The controller is free to make jumps at every time instant by switching $valve_i$, $i\in\{1,2\}$. Hence, the matrix coefficients in (\ref{E:AHS.2}) are set to zero according to the Remark~\ref{RM:1}. The matrices in (\ref{E:AHS.3}) are also obtained from the behavior of $\rho$ described above in this subsection and the fact that other state variables do not change at jumps.
The cost functional is selected as 
\begin{subequations}
\label{E:JCS}
\begin{align}
J &= \int_t^{t+T_h} l(\tau) d\tau + \sum_{i=1}^n [w_7+ (\rho_i^-)^2] \label{E:JCS1}\\
l &= w_4(P_{suc}-1.4)^2+w_5\sum_{i=1}^2(T_{a,i}-3)^2+w_6 u^2
\end{align}
\end{subequations}
from which the coefficients in (\ref{E:QJ}) can be determined.

For the simulations, the set of parameter values $T_h=200$ sec, $w_1=-0.1$, $w_2=0.1$, $w_3=200$, $w_4=2$, $w_5=0.2$, $w_6=10^{-3}$, $w_7=100$ is considered. The parameters heat flow to each display case $\dot{Q}_{airload}$ and miscellaneous refrigerant flow $\dot{m}_{ref,const}$ in \cite{R:SU2} are set to $3000$ J/sec and $0.2$ Kg/sec respectively. 
The simulation results for the indirect MPC method are shown in plots (a) through (d) of Fig.~\ref{FG:2}. The original nonlinear dynamical equations of the refrigeration system in \cite{R:SU2} are used for the simulation of the plant. The simulation time step is set to $T_s=1$ sec. The controller execution period is set to $T_c=5$ sec. 
The MLD-MPC approach is also applied to the refrigeration system using the Hybrid toolbox for MATLAB \cite{R:Bemporad2009}. 
The simulation results for the MLD-MPC method are shown in plots (e) through (h) of Fig.~\ref{FG:2}. 

\begin{figure}[!h]
\centering
\includegraphics[width=12cm]{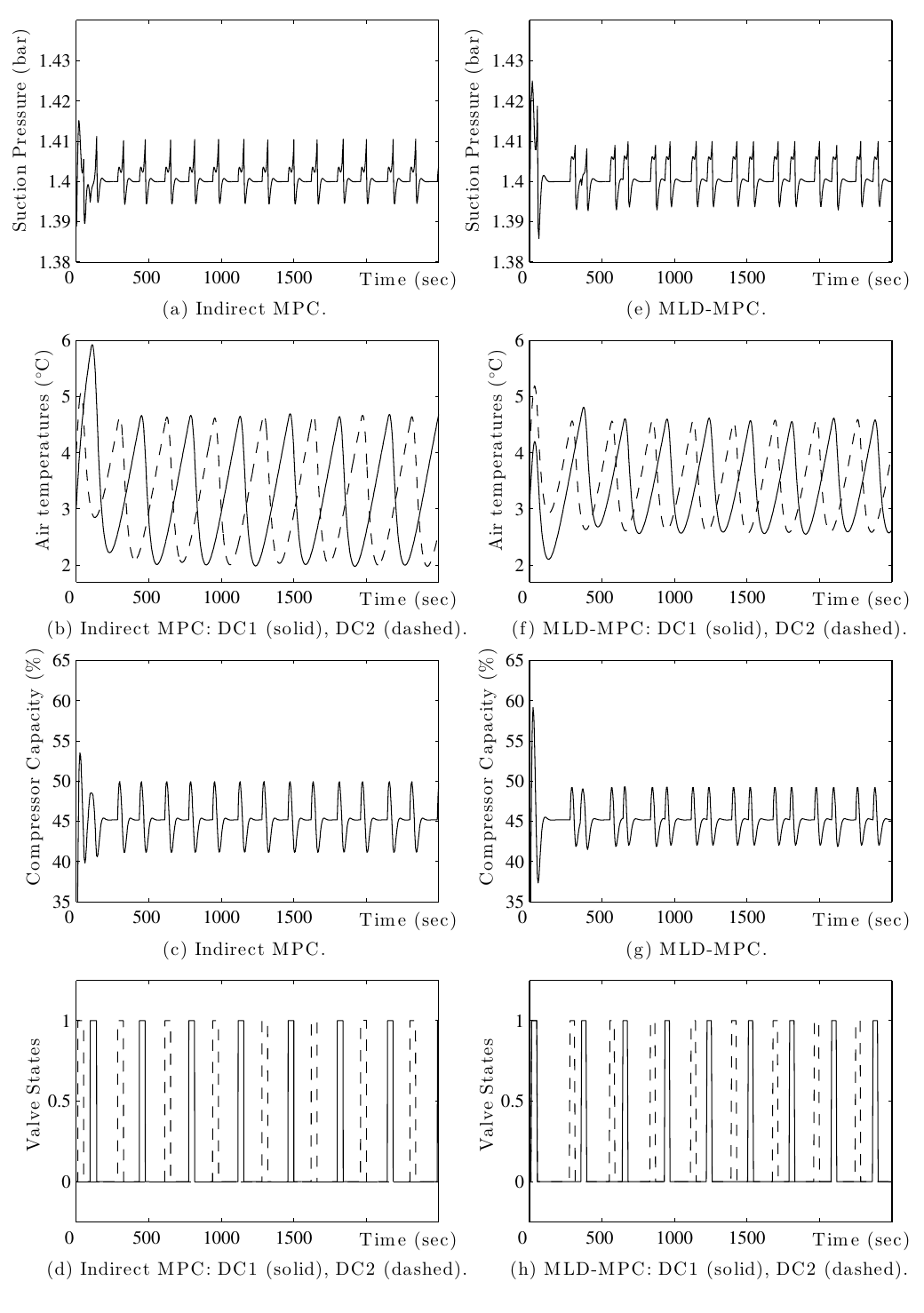}
\caption{Simulation of the indirect MPC and MLD-MPC methods for the refrigeration system with two display cases (denoted as DC1 and DC2).}
\label{FG:2}
\end{figure}

According to the Fig.~\ref{FG:2}, both of the MPC methods prevent from valve switching synchronization (that occur in the traditional controller). 
The closed loop time responses of the two methods are not exactly the same. The reason is that the calculations in MLD-MPC are in terms of the discretized time. But, the indirect MPC computes the optimal trajectories over the continuous time range of the prediction horizon which is more accurate. As a result, the valve switching times are distributed more evenly and regularly in the case of indirect MPC.

\subsection{Computational aspects}
The simulations of this section are carried out in the MATLAB\textsuperscript{\textregistered} 2017a environment on a PC with Intel\textsuperscript{\textregistered} Core\textsuperscript{TM} i7-4500U 1.8 GHz processor and 64 bit version of the Windows 7 operating system. The average and maximum values of controller execution times for simulation of 1000 controller execution steps are shown in the Table~\ref{TB:01} for the indirect MPC (Algorithm~\ref{ALG:1}) and the MLD-MPC. Several values of the prediction horizon $T_h$ and two different solver (for MLD-MPC) are considered. When $T_h$ increases to 100 sec, the average execution time of indirect MPC implemented in m-code becomes smaller than the average execution time of the MLD-MPC implemented using Gurobi v9.01 which is binary coded and is declared to be the fastest MIP solver \cite{R:Gurobi}. For $T_h = 200$ sec, the superiority of the indirect MPC is more than an order of magnitude. However, it is much more reasonable to compare the m-coded indirect MPC with MLD-MPC using miqp.m solver \cite{R:12} which is also written in m-code. In the case of this solver, no simulation progress was experienced after 3 hours for $T_h =200$ sec as indicated in the Table~\ref{TB:01}. 

The MLD model of the refrigeration system includes 7 real-valued auxiliary variables and 3 binary-valued auxiliary variables. As a result, the number of decision variables for the MIP which is solved at each step of the MLD-MPC method can be calculated as $7 n_h$ real plus $3 n_h$ binary variables with $n_h=T_h/T_c$. In general, a larger $T_h$ increases the number of MIP decision variables and the MLD-MPC execution time.
However, the controller execution time of indirect MPC is not always increasing with $T_h$ according to the Table~\ref{TB:01}. A considerable decrease occurs when moving from $T_h=50$ sec to $T_h=100$ sec. The reason is explained as follows.
The cost $J$ in \eqref{E:JC} can be divided into two parts: the cost of flows which involves the first term on the right side of \eqref{E:JC.1} and the cost of jumps (the remaining terms). 
The MPC controller forces the hybrid plant to make a jump if the difference in the cost of flows due to that jump is greater than the cost of that jump (i.e. a smaller $J$ is obtained by making that jump). If $T_h$ is too small, then the cost of flows will be also small such that its difference cannot become greater than the cost of a jump. Hence, the controller avoids jumps which results in loss of control and increase of the output errors along with time. Consequently, the optimal value of $J$ also increases at each time step. In summary, the control is lost if the cost function $J$ is selected inadequately. This argument is valid for both of the MPC methods. However, in the case of indirect MPC, the larger value of $J$ increases the number of iterations of the Algorithm~\ref{ALG:1} due to the Theorem~\ref{TH:5}. Such a condition happens for both $T_h = 50$ sec and $T_h=20$ sec which results in the larger controller execution times in the Table~\ref{TB:01}.

\begin{table}[!b]
\renewcommand{\arraystretch}{1.2}
\caption{Comparison of the controller execution times for the indirect MPC and the MLD-MPC.}
\label{TB:01}
\centering
\begin{tabular}{cccccc}
\hline
\hline
\multirow{2.4}{3.2cm}{MPC method} & \multirow{2.4}{1.2cm}{\centering } & \multicolumn{4}{c}{Prediction horizon $T_h$ (sec)} \\
\cmidrule{3-6}
 &  & 20 sec & 50 sec & 100 sec & 200 sec \\  
\hline
\hline
\multirow{2.4}{3.2cm}{Indirect~MPC (m-code)} & Avg. & 0.54 & 0.79 & 0.095 & 0.28 \\
\cmidrule{2-6}
& Max. & 3.67 & 4.02 & 1.03 & 9.4 \\
\hline
\multirow{2.4}{3.2cm}{MLD-MPC (miqp.m~solver~\cite{R:12})} & Avg. & 0.27 & 20.9 & 221.3 & $>$ 3 hours \\
\cmidrule{2-6}
& Max. & 0.38 & 35.2 & 269.9 & $>$ 3 hours \\
\hline
\multirow{2.4}{3.2cm}{MLD-MPC (Gurobi~solver~\cite{R:Gurobi})} & Avg. & 0.008 & 0.024 & 0.267 & 4.7 \\
\cmidrule{2-6}
 & Max. & 0.053 & 0.096 & 0.976 & 234.2 \\
\hline
\hline
\end{tabular}
\end{table}

The maximum and average values of some algorithm execution parameters during the simulations of the indirect MPC are also shown in the Table~\ref{TB:02}. These parameters include the number of algorithm iterations at each step, the final number of elements in the set $\mathcal{S}$ at each step, the total number of equations solved at each step (calls to either $\mathtt{JPMPa}$ or $\mathtt{JPMPb}$ functions), and the number of unknowns among all of the equations solved during the simulation. The fact that the load of indirect MPC increases if the control is lost (due to an inadequately selected cost function) also shows up in the all of the parameter values in the Table~\ref{TB:02}.

\begin{table}[!b]
\renewcommand{\arraystretch}{1.2}
\caption{Statistics of the iterations and variables of the Algorithm~\ref{ALG:1} during simulations.}
\label{TB:02}
\centering
\begin{tabular}{cccccc}
\hline
\hline
\multirow{2.4}{3.3cm}{Parameter} & \multirow{2.4}{1.2cm}{\centering } & \multicolumn{4}{c}{Prediction horizon $T_h$ (sec)} \\
\cmidrule{3-6}
 &  & 20 sec & 50 sec & 100 sec & 200 sec \\  
\hline
\hline
\multirow{2.4}{3.3cm}{Number of iterations} & Avg. & 4.4 & 5.3 & 1.8 & 3.2 \\
\cmidrule{2-6}
& Max. & 13 & 13 & 3 & 15 \\
\hline
\multirow{2.4}{3.3cm}{Size of the set $\mathcal{S}$} & Avg. & 9.8 & 11.7 & 4.6 & 7.4 \\
\cmidrule{2-6}
& Max. & 27 & 27 & 7 & 31 \\
\hline
\multirow{2.4}{3.3cm}{Number of equations solved} & Avg. & 35.4 & 47 & 6.7 & 19.5 \\
\cmidrule{2-6}
& Max. & 205 & 181 & 15 & 219 \\
\hline
\multirow{2.4}{3.3cm}{Number of equation unknowns} & Avg. & 1.76 & 1.8 & 1.27 & 1.5 \\
\cmidrule{2-6}
& Max. & 5 & 5 & 2 & 4 \\
\hline
\hline
\end{tabular}
\end{table}

To demonstrate an application of the Theorem~\ref{TH:5}, it is considered that the MPC algorithm performs better than the traditional controller in reducing the value of its cost functional $J$. Hence, the worst case value of $J$ obtained from a simulation of the traditional controller which is 784 for $T_h = 100$ sec is used for applying the theorem. The value of $n_a$ in (\ref{E:NSB2}) is calculated as $n_a=2$. We also have $h_{min}=w_7=100$ which gives an upper bound for the number of iterations of the Algorithm~\ref{A:1} as $2^{1+\lfloor 784/100\rfloor}-1=255$. For $T_h = 200$ sec, the upper bound on the number of iterations increases to $2^{1+\lfloor (2\times 784)/100\rfloor}-1=65535$.

\section{Conclusion}
\label{S:CN}
The main existing approach to the hybrid MPC uses a direct approach to solve the finite horizon optimal control in the MPC setup. It converts the problem to a mixed integer program with possibly a large number of decision variables. In this work, an MPC method was proposed based on the indirect solution approach using the extended version of the Pontryagin's maximum principle for hybrid systems. The central part of the method is an algorithm which iterates on the sequences of discrete state and discrete input values in order to compute the optimal inputs at every time step. The computations are reduced to solving an algebraic system of equations for the case of affine hybrid systems. The algorithm is guaranteed to terminate in a finite number of steps if the cost functional of the MPC assigns cost to the jumps. The proposed approach was applied to a benchmark hybrid system control problem as a case study during which comparisons were made with the main existing hybrid MPC method. The results verify the superior performance of the proposed MPC method, especially for larger values of the prediction horizon. 
It is expectable that the numerical efficiency of the current initial implementation of the proposed MPC method which is based on the MATLAB m-code language can be furtherly improved during the future works.
Several other issues, including stability analysis, handling of state and input constraints, and application of the method to more case studies are subjects for the future works.

\appendix
\section{Appendix A: Calculation of $\mathcal{Y}_a$ and $\mathcal{Y}_b$}
\label{AX:A}
In this appendix, a technique is proposed for reducing the dimensionality of equations that should be solved in the functions $\mathtt{JPMPa}$ and $\mathtt{JPMPb}$ for calculation of $\mathcal{Y}_a$ and $\mathcal{Y}_b$ in (\ref{E:FS.1}) and (\ref{E:FS.1m}) from $t^s_1,\cdots,t^s_{n-1}$. For brievity, a jump $(q_i,\sigma_i,q_{i+1})\in \Theta$ which appears as a subscript index of a matrix coefficient is replaced by $i$. Also, $\Psi_{q_i}(t^s_i-t^s_{i-1})$ in (\ref{E:LP.1}) is briefly denoted as $\Psi_i$. One can use the Equation (\ref{E:2.2}) with (\ref{E:AHS.3}) and the fact that $x_0^+=x_{ic}$ to write the following equations.

\begin{subequations}
\label{E:AL}
\begin{align}
\begin{bmatrix}x_0^+ \\ \lambda_0^+ \\ 1 \end{bmatrix} &= \Lambda_0 \begin{bmatrix}\lambda_0^+\\1\end{bmatrix}, 
\qquad 
\Lambda_0 = \begin{bmatrix}0&x_{ic}\\I&0\\0&1\end{bmatrix}
\label{E:AL0}
\\
\begin{bmatrix}x_i^+\\\lambda_i^+\\1\end{bmatrix} &= \Lambda_i \begin{bmatrix}x_i^-\\\lambda_i^+\\1\end{bmatrix},
\qquad \underset{\displaystyle{i\in[1..n-1]},\hfill}{\Lambda_i = \begin{bmatrix} L^x_{i}&0&L^c_{i} \\ 0&I&0 \\ 0&0&1 \end{bmatrix}}
\label{E:AL1}
\end{align}
\end{subequations}

Also, defining $\hat{M}^x_i$ such that $\hat{M}^x_i {M^x_{i}}^T = 0$, one can use (\ref{E:PMPj.1}), (\ref{E:GB}), (\ref{E:AHS}), and (\ref{E:QJ}) to write the following equations for every $i\in[1..n-1]$.

\begin{subequations}
\label{E:AO}
\begin{align}
&\Omega_i \begin{bmatrix} x_i^- \\ \lambda_i^+ \\ 1 \end{bmatrix} = \Pi_i \begin{bmatrix} x_i^- \\ \lambda_i^- \\ 1 \end{bmatrix}, \quad
\Pi_i = \begin{bmatrix} I&0&0 \\ 0&\hat{M}^x_i&0 \\ 0&0&0\\0&0&1\end{bmatrix} 
\label{E:AO2} \\
&\Omega_i = \begin{bmatrix} I&0&0 \\ \hat{M}^x_i W^{jx}_{i} & \hat{M}^x_i {L^x_{i}}^T & -\hat{M}^x_i W^{jx}_{i} \bar{x}_{q_i} \\ M^x_{i}&0&M^c_{i}\\0&0&1\end{bmatrix} 
\label{E:AO3}
\end{align}
\end{subequations}

Also, (\ref{E:PMPf.1}) is written as

\begin{align}
\label{E:AF}
\Omega_e \begin{bmatrix} x_n^- \\ \lambda_n^- \\ 1 \end{bmatrix} = 0, \quad \Omega_e = \begin{bmatrix}W^f_{q_n}&-I&-W^f_{q_n}\bar{x}_{q_n}\end{bmatrix} 
\end{align}

If $n=1$, then (\ref{E:AL0}), (\ref{E:LP.1}) for $i=1$, and (\ref{E:AF}) can be combined as $\Omega_e \Psi_1 \Lambda_0$ $ [\lambda_0^{+^T}~1]^T=0$. This equation can be solved for $\lambda_0^+$ from which $x_1^-$ and $\lambda_1^-$ are obtained using (\ref{E:LP.1}) for $i=1$. Otherwise, if $n>1$, one can relplace (\ref{E:AL}) in (\ref{E:LP.1}) and replace the result in (\ref{E:AO3}) and (\ref{E:AF}) to obtain the following equations. 

\begin{subequations}
\label{E:AR}
\begin{align}
&\Omega_1 \begin{bmatrix} x_1^- \\ \lambda_1^+ \\ 1 \end{bmatrix} = \Pi_1 \Psi_1 \Lambda_0 \begin{bmatrix}\lambda_0^+\\1\end{bmatrix} 
\label{E:AR1} \\
&\Omega_i \begin{bmatrix} x_i^- \\ \lambda_i^+ \\ 1 \end{bmatrix} = \Pi_i \Psi_i \Lambda_{i-1} \begin{bmatrix}x_{i-1}^-\\\lambda_{i-1}^+\\1\end{bmatrix}, \quad i\in[2..n-1] 
\label{E:AR2} \\
&\Omega_e \Psi_n \Lambda_{n-1} \begin{bmatrix}x_{n-1}^-\\\lambda_{n-1}^+\\1\end{bmatrix} = 0 
\label{E:AR3}
\end{align}
\end{subequations}

The number of rows in (\ref{E:AR1}) and (\ref{E:AR2}) is $2n_x+1$, where the equations given by the last rows are trivial. Hence, (\ref{E:AR}) is a system of $(2n-1)n_x$ linear equations in terms of $(2n-1)n_x$ unknowns in $\lambda_i^+ (0\le i<n)$ and $x_i^- (0<i<n)$. After solving this system of equations, $\lambda_i^- (0<i\le n)$, $x_i^+ (0<i<n)$, and $x_n^-$ can be obtained using (\ref{E:LP.1}), (\ref{E:AL1}), and (\ref{E:AF}) respectively. 

\section{Appendix B: Calculation of the cost functional}
\label{AX:B}
In the functions $\mathtt{JPMPa}$ ($\mathtt{JPMPb}$), the cost functional $J$ ($J_m$) should be calculated given $\mathcal{Y}_a$ ($\mathcal{Y}_b$).
In this appendix a method is proposed for calculating the part of $J$ that involve integrations on the right hand side of (\ref{E:JC.1}) denoted as $J_1$ in the following (other parts are already in terms of the elements in $\mathcal{Y}_a$ or $\mathcal{Y}_b$).  

\begin{align}
J_1= {\frac{1}{2}} \sum_{i=1}^n \int_{t^s_{i-1}}^{t^s_i} \underset{\displaystyle(u-\bar{u}_{q_i})^T W^u_{q_i} (u-\bar{u}_{q_i})+ W^c_{q_i}]dt}{[(x-\bar{x}_{q_i})^T W^x_{q_i} (x-\bar{x}_{q_i})+\hfill}
\end{align}

By defining $z^T(t)=[x^T(t)~\lambda^T(t)~1]$ and using (\ref{E:US}), the above equation together with (\ref{E:LP.4}) can be transformed into the following form.

\begin{subequations}
\label{E:BT}
\begin{align}
J_1&={\frac{1}{2}}\sum_{i=1}^n J_{1,i} 
\label{E:BT1} \\
J_{1,i} &= \int_{t^s_{i-1}}^{t^s_i} z^T(t) W^z_i z(t) dt, \quad i\in[1..n]  
\label{E:BT2} \\
\dot{z} &= A^e_{q_i} z, \quad t\in[t^s_{i-1},t^s_i), \quad i\in[1..n] 
\label{E:BT3} \\
W^z_i &= \begin{bmatrix} W^x_{q_i}&0&-W^x_{q_i}\bar{x}_{q_i} \\ 0&B^u_{q_i}W^{u^{-1}}_{q_i}{B^u_{q_i}}^T&0 \\ -\bar{x}_{q_i}^T W^x_{q_i}&0&W^c_{q_i}+\bar{x}_{q_i}^T W^x_{q_i}\bar{x}_{q_i} \end{bmatrix}
\label{E:BT4}
\end{align}
\end{subequations}

According to (\ref{E:BT3}) one can write $z(t)=e^{-A^e_{q_i}(t^s_i-t)} z_i^-$ which can be replaced in (\ref{E:BT2}) to obtain

\begin{align}
J_{1,i} &= z_i^{-^T} \int_{t^s_{i-1}}^{t^s_i} e^{-A^{e^T}_{q_i} (t^s_i-t)} W^z_i z(t) dt 
\end{align}

The above equation can be written as (\ref{E:BC1}) in the following, in which $\hat{z}$ evolves according to the differential equation $\dot{\hat{z}} = -A^{e^T}_{q_i} \hat{z} + W^z_i z$ with initial conditions $\hat{z}_{i-1}^+=0$. This differential equation together with (\ref{E:BT3}) can be represented as (\ref{E:BC2}) which is solved as (\ref{E:BC3}). 

\begin{subequations}
\label{E:BC}
\begin{align}
J_{1,i} &= z_i^{-^T} \hat{z}_i^- =
\begin{bmatrix} z_i^-\\ \hat{z}_i^- \end{bmatrix}^T \begin{bmatrix} 0&I \\ 0&0 \end{bmatrix} \begin{bmatrix} z_i^-\\ \hat{z}_i^- \end{bmatrix}
\label{E:BC1}\\
\frac{d}{dt}\begin{bmatrix} z\\ \hat{z} \end{bmatrix} &= A^g_{q_i} \begin{bmatrix} z\\ \hat{z} \end{bmatrix}, \quad A^g_{q_i} =\begin{bmatrix} A^e_{q_i}&0 \\ W^z_i&-A^{e^T}_{q_i} \end{bmatrix} 
\label{E:BC2}\\
\begin{bmatrix} z_i^-\\ \hat{z}_i^- \end{bmatrix} &= \Psi^g_i \begin{bmatrix} z_{i-1}^+\\ 0 \end{bmatrix}, \quad \Psi^g_i = e^{A^g_{q_i}} 
\label{E:BC3}
\end{align}
\end{subequations}

Replacing (\ref{E:BC3}) in (\ref{E:BC1}), $J_{1,i}$ for $i\in[1..n]$ is computed as in (\ref{E:BCr}).

\begin{subequations}
\label{E:BCr}
\begin{align}
J_{1,i} &= z_i^{-^T} \Psi^g_{i,21} z_{i-1}^{+^T}
\label{E:BC4}\\
\Psi^g_{i,21} &= \begin{bmatrix} 0&I \end{bmatrix} \Psi^g_i \begin{bmatrix} I \\ 0 \end{bmatrix}
\label{E:BC5}
\end{align}
\end{subequations}

\end{document}